\documentclass[letterpaper, 10 pt, conference]{ieeeconf}

\IEEEoverridecommandlockouts  

\usepackage[framemethod=tikz,tikzsetting={draw=gray,dashed,line width=2pt,dash pattern = on 10pt off 3pt},
linecolor=white,
backgroundcolor=yellow!30,
outerlinewidth=1pt]{mdframed}
\usepackage{xcolor}

\usepackage{hyperref}
\hypersetup{
    colorlinks = true,
    linkcolor=.,
    anchorcolor=.,
    citecolor={black},
    filecolor=.,
    menucolor=.,
    runcolor=.,
    urlcolor=.,
}

\usepackage{booktabs} 
\usepackage[utf8]{inputenc}
\usepackage{amsmath}
\usepackage{amsfonts}
\usepackage{amssymb}
\usepackage{float}
\usepackage{graphicx}
\usepackage{amsthm}
\usepackage{pdfpages}
\usepackage{pgfplots}
\usepackage{tikz}
\usetikzlibrary{arrows,shapes,decorations,backgrounds}
\usepackage{verbatim}
\usepackage{float}
\usepgfplotslibrary{fillbetween}
\usepackage{subcaption}
\usepackage[font=footnotesize]{caption}
\usepackage{color}

\theoremstyle{definition}
\newtheorem{proposition}{Proposition}
\theoremstyle{definition}

\theoremstyle{definition}

\theoremstyle{definition}
\newtheorem{remark}{Remark}
\theoremstyle{definition}

\theoremstyle{definition}
\theoremstyle{definition}
\newtheorem{theorem}{Theorem}
\theoremstyle{definition}
\newtheorem{assumption}{Assumption}
\theoremstyle{definition}

\theoremstyle{definition}

\theoremstyle{definition}

\usepackage[
language=english,
style=ieee,
 backend=bibtex,
 bibencoding=utf8,
 giveninits=true,
 backref=false,
  maxnames=3,
 maxcitenames=3,
 doi=false,isbn=false,url=false,
 ]{biblatex}
\title{\LARGE \bf Adaptive Control for Flow and Volume Regulation in Multi-Producer District Heating Systems}

\author{Juan E. Machado, Michele Cucuzzella, Nina Pronk, and Jacquelien~M.~A.~Scherpen
	\thanks{J. E. Machado, M. Cucuzzella, N. Pronk and J. M. A. Scherpen are with the Jan C. Willems Center for Systems and Control, ENTEG, Faculty of Science and
Engineering, University of Groningen, Nijenborgh 4, 9747 AG Groningen, the
Netherlands (email: \{j.e.machado.martinez, j.m.a.scherpen\}@rug.nl, n.s.pronk@student.rug.nl). M. Cucuzzella is also with the Department of Electrical, Computer and Biomedical Engineering, University of Pavia, via Ferrata 5, 27100 Pavia, Italy (email: michele.cucuzzella@unipv.it).}
}%

\begin{document}

\maketitle

\begin{abstract}
Flow and storage volume regulation is essential for the adequate transport and management of energy resources in district heating systems.  In this letter, we propose a  novel and  suitably tailored---decentralized---adaptive control scheme addressing this problem whilst offering closed-loop stability guarantees. We focus on a system configuration comprising multiple heat producers, consumers and storage tanks exchanging energy through a common distribution network, which are features of modern and prospective district heating installations. {\color{black}The proposed controller is based on passivity, backstepping and (indirect) adaptive control theory.}
\end{abstract}

\section{Introduction}


District Heating (DH) comprises a network of insulated pipes which  transport heated fluid, carrying thermal power from heating stations (producers) towards clusters of consumers within a neighborhood, town center or city \cite{Lund2014}. To  further  unlock the potential of these systems for a more sustainable heating sector, prospective installations will substantially increase the share of renewable energy sources (e.g., geothermal or solar thermal), waste heat sources from industrial or commercial buildings, as well as thermal storage units, promoting as a consequence DH installations featuring  multiple---potentially distributed---heat producers  and distribution networks of meshed topology \cite{Lund2014}, \cite{Dominkovic2017}, \cite{vesterlund_optim_2017}, \cite{wang_meshed_17}.

The effective distribution of heat and  management of energetic resources in DH systems strongly depends on the adequate regulation of the system temperatures, pressures and flows \cite{Werner2017}, \cite{vandermeulen_control_review_18}.  The design of control systems to tackle these objectives is significantly challenging due to the nonlinear, networked, and uncertain nature of DH systems models \cite{Scholten_tcst_2015}, \cite{DePersis2011}. The control of  DH systems with a \emph{single} heat producer has received considerable attention. In \cite{DePersis2014} global asymptotic end-user (consumer) pressure regulation was addressed via decentralized proportional-integral controllers (c.f., \cite{DePersis2011}), whereas temperature and storage volume regulation was achieved in \cite{Scholten_tcst_2015} via a novel internal model controller. Temperature control was also investigated using Lyapunov-Krasovskii theory  in \cite{bendtsen_control_2017} for a system model which accounts for  non-negligible  delays in the heat transport from producer and consumers. For the case of DH system with \emph{multiple} heat producers, a number of works have focused on design and operational optimization, e.g., \cite{Dominkovic2017}, \cite{vesterlund_optim_2017} and \cite{wang_optimization_2017}. The use of predictive control was investigated in \cite{sandou_predictive_05}  for optimal system operation. However, the implementation requires system-wide measurements (of temperatures) and no formal stability analysis is presented. The optimal regulation of flow networks was addressed in \cite{Trip2019a}, including as a particular case a class of simplified DH systems with storage units. Even though closed-loop stability is guaranteed under some conditions, the system model neglects friction effects on pipelines, which are significant in these applications.

In this letter, we propose   a novel adaptive controller for flow and storage volume regulation of a \emph{multi}-producer DH system. {\color{black}The control scheme is   decentralized, as  each  control input depends only on locally available information, and is based on backstepping,  adaptive and passivity-based control design tools.}  {\color{black}We consider a nonlinear and uncertain system model which accounts the effects associated with friction in pipes and is suitable to describe general distribution network topologies.} Conditions for closed-loop asymptotic stability are also presented.


{\bf Notation:}  $\mathbb{R}$ denotes the set of real numbers. For a vector $x\in\mathbb{R}^n$,  $x_i$ represents its $i$th component, {i.e.}, $x=[x_1,\dots, x_n]^\top$ and   $\vert x \vert = [\vert x_1 \vert,\dots, \vert x_n \vert]^\top$. An $m\times n$ matrix with all-zero entries is written as $\boldsymbol{0}_{m\times n}$. An $n$-vector of ones is written as $\boldsymbol{1}_n$, whereas the identity matrix of size $n$ is represented by $I_n$. For any vector $x\in \mathbb{R}^n$, we denote by $\langle x \rangle$  a diagonal matrix with elements $x_i$ in its main diagonal. For any time-varying signal $w$, we represent by $\bar w$ its steady-state value, if exists.

\section{Background and problem formulation}\label{sec:2}

In this section we introduce the  multi-producer DH system under consideration and formulate the flow and storage volume regulation control problem.  

\subsection{System model}

\begin{figure}[t]
\begin{center}
\includegraphics[width=0.8\linewidth]{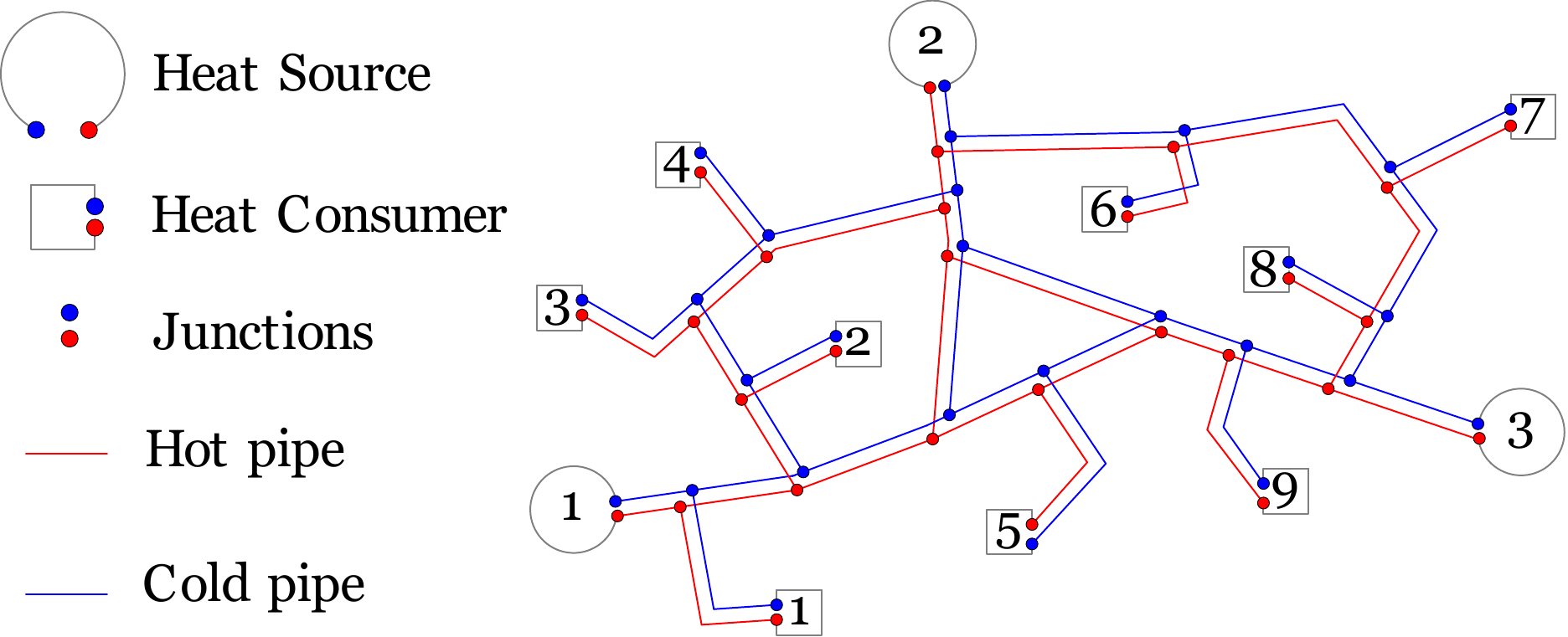}
\caption{Sketch of a simplified DH system (c.f.  \cite{wang_meshed_17}).}
\label{fig:1}
\end{center}
\end{figure}

We consider the hydraulic system of a water-based (leak free) DH system comprising $n_\mathrm{pr}$ heat producers, $n_\mathrm{c}$ consumers and $n_\mathrm{ST}$ storage tanks which are  connected to a common distribution network. The latter is assumed to be symmetric in the sense that  supply and return layers, which respectively transport hot and cold water,  have the same topology. In  Fig.~\ref{fig:1} a simplified DH system with three producers and nine consumers is shown.\footnote{\color{black}It is assumed that water is incompressible and that its density $\rho$ is constant. All system pipes are assumed to be cylindrical.} 

In this work,   producers, consumers and distribution network are assumed to be composed of elementary hydraulic devices, namely,  valves, pipes and pumps. {\color{black}Producers are assisted by hydraulic pumps to deliver thermal power to the system by circulating and heating water through heat exchangers (viewed here as pipes): cold water is continuously drawn from the return layer of the distribution network which is then  heated and injected back into the supply layer.} The operation mode of consumers is analogous to that of producers.  Storage tanks  accumulate volumes of hot and cold water  which are perfectly separated by a  thermocline, i.e.,  hot water is on top and cold water at the bottom, and without heat exchange between them. In addition, each tank is considered to have four valves, two at the top and two at the bottom, which are used as inlets and outlets of hot and cold water, respectively.  The foregoing description is schematically depicted  in Fig.~\ref{fig:2}; for more details see \cite{Scholten_tcst_2015}.

\begin{figure}[t]
\begin{center}
\includegraphics[width=0.8\linewidth]{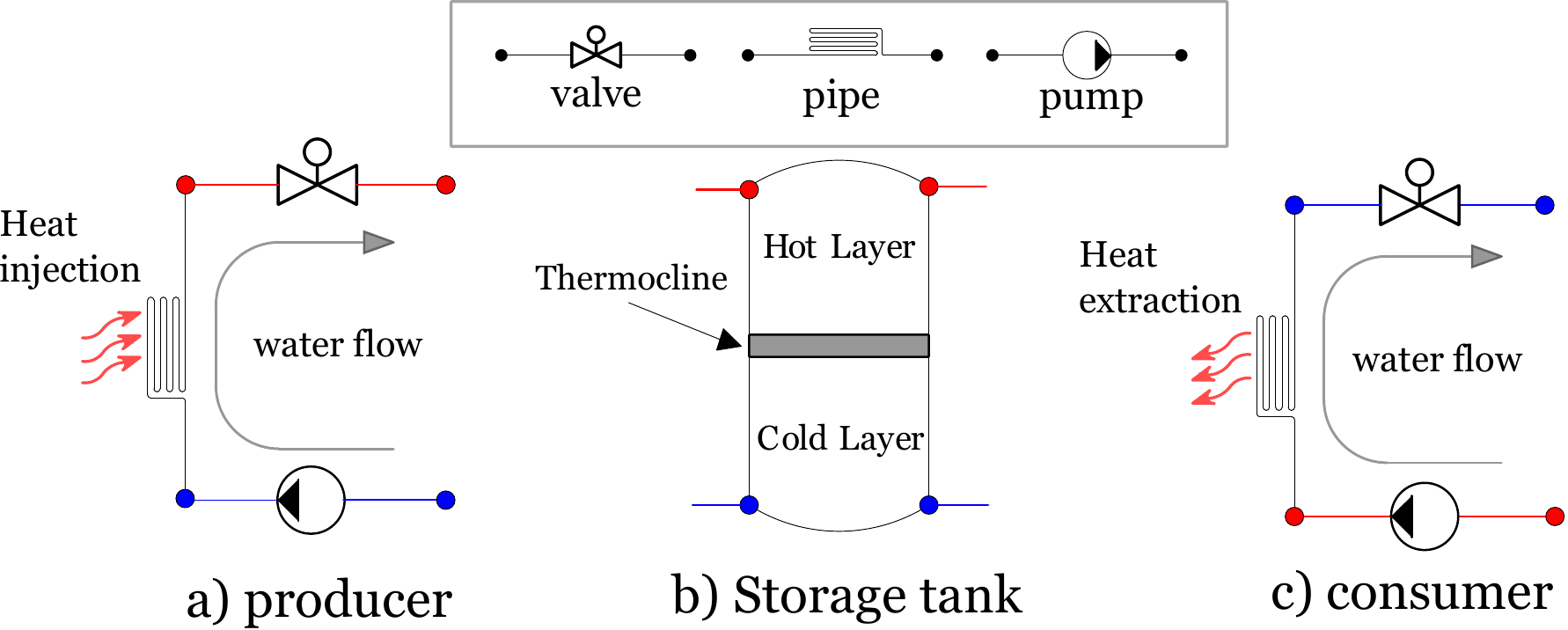}
\caption{Topologies of producers, consumers and storage tanks; see \cite{Scholten_tcst_2015,DePersis2011}. Pipes of producers and producers  represent heat exchangers.}
\label{fig:2}
\end{center}
\end{figure}

The DH system is henceforth viewed as the connected graph $\mathcal{G}=(\mathcal{N},\mathcal{E})$. The nodes $\mathcal{N}$ are all the system junctions as well as the hot and cold layers of the storage tanks. All two terminal devices{\color{black}, namely, pumps, pipes and valves,} are represented by the set of edges $\mathcal{E}$. {\color{black}Each edge is assumed to have an arbitrary and fixed orientation and this  is codified through the node-edge incidence matrix $\mathcal{B}_0$. }  For any edge $i\in\mathcal{E}$,  $q_{\mathrm{E},i}$ and $V_{\mathrm{E},i}$ are the flow through it and the volume of water in it, respectively. Also, $V_{\mathrm{N},k}$ and $p_{\mathrm{N},k}$ are the volume and pressure of a given node $k\in\mathcal{N}$. {\color{black} The cardinalities of $\mathcal{E}$ and $\mathcal{N}$ are denoted by $n_\mathrm{E}$ and $n_\mathrm{N}$, respectively.}  

Basic models describing the dynamic behavior of the flows $q_{\mathrm{E},i}$ and the volumes $V_{\mathrm{N},k}$ are presented next.

\subsubsection{Dynamics of Edges}

For every edge $i\in \mathcal{E}$, with terminals $j,k\in\mathcal{N}$, we relate the rate of change of the flow through it with the pressure drop across it via the differential-algebraic equation \cite{DePersis2011}
\begin{equation}\label{eq:dae_flows}
p_{\mathrm{N},j}-p_{\mathrm{N},k}=J_{\mathrm{E},i}\dot{q}_{\mathrm{E},i}+f_{\mathrm{E},i}(q_{\mathrm{E},i})-w_{\mathrm{E},i},
\end{equation}
{\color{black}where the nodes $j\neq k$ are the endpoints of edge $i$.} In \eqref{eq:dae_flows},  if $i$ represents a pipe, then  $J_{\mathrm{E},i}>0$ is a constant depending on its physical dimensions, {\color{black} such as length and cross-sectional area}. {\color{black}If $i$ is a pump, then $w_i$ denotes the pressure difference that $i$ produces  across its terminals.} The function $f_{\mathrm{E},i}$ is  assumed to be continuously differentiable, mononotically increasing and its explicit form may be unknown or depend on uncertain parameters. {\color{black}Indeed, for a given pipe $i\in \mathcal{E}$, the function $f_{\mathrm{E},i}$ models the pressure drop caused by the friction between the pipe's interior wall and the stream of water;  through $f_{\mathrm{E},i}$ we also model the pressure drop caused by valves.

In Section~\ref{sec:reduced_hydr} we recall results from \cite{jmc_dh_modeling_2020} to associate \eqref{eq:dae_flows} to an equivalent ODE-based model. Consider first the following additional details on $f_{\mathrm{E},i}$:

\begin{assumption}\label{assu:f_E_pipe}
For each $i\in \mathcal{E}$ representing a pipe or a valve, the function $f_\mathrm{E,i}$ in \eqref{eq:dae_flows} is given by
\begin{equation}\label{eq:model_f_E}
f_{\mathrm{E},i}(q_{\mathrm{E},i})=\theta_i \vert q_{\mathrm{E},i}\vert q_{\mathrm{E},i},
\end{equation}
where $\theta_i$ is an {\em unknown} positive scalar. If $i\in \mathcal{E}$ is a pump, then we take $\theta_i=0$.
\end{assumption}
}

{\color{black}
\begin{remark}
For pipes, the model \eqref{eq:model_f_E} is related to the Darcy-Weisbach formula and is widely used in the literature (see, {e.g.}, \cite{wang_meshed_17, grosswindhager_delay_11,palsson_book_99,Hauschild2020}) to describe pressure drops in hydraulic systems  due to the effects of friction in pipes.  The coefficient $\theta_i$, which we treat as an unknown positive scalar,  is related to the pipe's length $\ell_i$, internal diameter $d_i$ and the  friction factor $k_i$ through the expression $\theta_i=(k_i\ell_i \rho)/(2d_i)$. The values of $\ell_i$ and $d_i$ may be accurately known (or measured). However, the friction factor $k_i$ depends on the pipe's roughness $\epsilon_i$, which is difficult to measure or whose measurement may considerable deviate from its true value due to, e.g., aging or corrosion. The friction factor $k_i$ is related to $\epsilon_i$  through the Colebrook equation, given by \cite{cengel_thermo_08}
\begin{align*}
\tfrac{1}{\sqrt{k_i}} = -2\log_{10}\left(\tfrac{\epsilon_i/d_i}{3.7}+\tfrac{2.51}{\mathrm{Re}_i  \sqrt{k_i}} \right);~~~\mathrm{Re}_i = \tfrac{\rho \vert q_{\mathrm{E},i} \vert }{(\pi d_i/4) \nu},
\end{align*} 
where  $\mathrm{Re}_i$  is the Reynolds number for turbulent flow. We note that  $\mathrm{Re}_i$ depends on the flow $q_{\mathrm{E},i}$ and the viscosity of water $\nu$, which in turn depends on the fluid's temperature. In this work, out of simplicity we consider that $\mathrm{Re}_i$ is constant, with its value corresponding to that of nominal operating flow and temperature conditions.
\end{remark}
}

{\color{black}
\begin{remark}
For any valve we also use the model \eqref{eq:model_f_E} to describe pressure drop through it as done, {e.g.}, in \cite{valdimarsson_14} and \cite{wang_meshed_17}. However, in this case, the coefficient $\theta_i$ is related with quantities that depend on the characteristic of the specific type of valve, e.g., the {\em rangeability} for equal percentage valves, or other generic parameters such as the (maximum) flow capacity and the opening degree of the valve. Even if the values of some of these parameters are known (from manufacturers data sheet, for example), diverse degradation mechanisms such as corrosion can change these parameters' values over the course of a valve service period. Therefore, $\theta_i$ is also considered to be an unknown positive scalar for valves.  
\end{remark}
}

\subsubsection{Dynamics of Nodes}

{\color{black}
The volume $V_{\mathrm{N},k}$ of every  node $k\in\mathcal{N}$ in the DH system  evolves according to the mass balance equation per node, which in view of the assumption of incompressibility and constant density of water,  is equivalent to  $\dot{V}_{\mathrm{N},k}=\sum_{i\in\mathcal{I}_k}q_{\mathrm{E},i}$, 
where $\mathcal{I}_k$ is the set of edges that are incident to node $k$. Considering the DH system's incidence matrix $\mathcal{B}_0$, the set of all these equations, for all $k\in \mathcal{N}$, can be written as 
\begin{equation}\label{eq:vol_dyn_NR}
\dot{V}_\mathrm{N}=\mathcal{B}_0 q_\mathrm{E},
\end{equation}
where $q_\mathrm{E}\in \mathbb{R}^{n_\mathrm{E}}$ is a vector comprising the flow through every edge in the DH system. We assume that for any $k\in\mathcal{N}$ representing a simple junction,  $V_{\mathrm{N},k}=\delta_k$ for all time, with $\delta_k\geq 0$ constant. Then,  \eqref{eq:vol_dyn_NR} becomes a differential-algebraic equation.  The latter assumption, which stems from the fact that simple junctions are of a much smaller dimension than storage tanks, is useful to write a reduced-order, ODE-based model equivalent to \eqref{eq:vol_dyn_NR} that focuses on the volume dynamics of storage tanks; see Section~\ref{sec:reduced_hydr}. Out of simplicity we take $\delta_k=0$.
}

\subsubsection{A reduced-order ODE-based hydraulic model}\label{sec:reduced_hydr}

{\color{black}
We recall first three instrumental assumptions from \cite{jmc_dh_modeling_2020} to write equivalent ODE-based models of \eqref{eq:dae_flows} and \eqref{eq:vol_dyn_NR}. These are: {\bf (a)} every producer is interfaced to the distribution network through a storage tank as depicted in Fig.~\ref{fig:3}; {\bf (b)}~the total volume of water in each tank remains constant and at maximum capacity for all time; and {\bf (c)} there are no standalone storage tanks in the system. The reader is referred to Section~2.1 of \cite{jmc_dh_modeling_2020} for details.
}

{\color{black}
Regarding the flow dynamics \eqref{eq:dae_flows}, we summarize the results in \cite{jmc_dh_modeling_2020} as follows. Considering assumptions {\bf (a)}, {\bf (b)} and {\bf (c)}, the overall DH system's flow vector $q_\mathrm{E}$ is completely determined by the flow through a selected number of devices. More precisely, there exists a constant matrix $\mathcal{F}$ of the form
\begin{equation}\label{eq:fundamental_loop}
\mathcal{F}=\begin{bmatrix}
I_{n_\mathrm{ch}} & 0 & G & 0 \\
0 & I_{n_\mathrm{pr}} & 0 & H
\end{bmatrix},~ G_{ij},H_{ij}\in \{-1,0,1 \},
\end{equation} 
such that
\begin{equation}\label{eq:q_E}
q_\mathrm{E}=\mathcal{F}^\top \begin{bmatrix}
q_\mathrm{ch}\\
q_\mathrm{pr}
\end{bmatrix},
\end{equation}
where $q_\mathrm{pr}\in \mathbb{R}^{n_\mathrm{pr}}$ comprises the flows through each producer and  $q_\mathrm{ch}\in \mathbb{R}^{n_\mathrm{ch}}$  stacks the flows through each consumer, through some pipes of the distribution network (one per loop), and  the flow at the hot layer's outlet pipe of each storage tank (except for one). Then,  $n_\mathrm{ch}=n_\mathrm{c}+a+n_\mathrm{pr}-1$, where $a\geq 0$ denotes the number of loops in the distribution network.\footnote{{\color{black}The elements of $q_\mathrm{ch}$ and $q_\mathrm{pr}$ represent a set of independent variables. In the example of Fig.~\ref{fig:3}, $q_\mathrm{ch}=(q_{\mathrm{E},1},q_{\mathrm{E},2})$, $q_\mathrm{pr}=(q_{\mathrm{E},3},q_{\mathrm{E},4})$ and $a=0$. The remaining flows are dependent on these as $q_{\mathrm{E},6}=q_{\mathrm{E},2}-q_{\mathrm{E},1}$ and, due to assumption~{\bf(b)},  $q_{\mathrm{E},5}=q_{\mathrm{E},2}$ and $q_{\mathrm{E},7}=q_{\mathrm{E},6}$.}}

Moreover, $q_\mathrm{ch}$ and $q_\mathrm{pr}$ satisfy the following set of decoupled ODEs:
\vspace{-2.5mm}
\begin{subequations}\label{eq:flow_ODE}
\begin{align}
J_\mathrm{ch}\dot{q}_\mathrm{ch} & = f_\mathrm{ch}(q_\mathrm{ch})+u_\mathrm{ch}, \label{eq:flow_ODE_qch}\\
J_\mathrm{pr}\dot{q}_\mathrm{pr} & = f_\mathrm{pr}(q_\mathrm{pr})+u_\mathrm{pr}, \label{eq:flow_ODE_qpr}
\end{align}
\end{subequations}
where the matrix $J_\mathrm{ch}$ is symmetric, positive definite, and $J_\mathrm{pr}$ is diagonal and positive definite as well; both matrices are constant and depend on the parameters $J_{\mathrm{E},i}$ in \eqref{eq:dae_flows}.  Also,  $u_{\mathrm{ch},i}$ and $u_{\mathrm{pr},i}$ are independent control inputs and represent the pressure difference of hydraulic pumps in series with the devices whose flows are $q_{\mathrm{ch},i}$ and $q_{\mathrm{pr},i}$, respectively. Moreover,  $-f_\mathrm{ch}$ and $-f_\mathrm{pr}$, which are associated with $f_{\mathrm{E},i}$ in \eqref{eq:dae_flows} (see also \eqref{eq:model_f_E}), are nonlinear, continuously differentiable and {\em monotone} mappings. Notably, each component of $f_{\mathrm{pr},i}$ can be written as
\begin{equation}\label{eq:f_pr}
f_{\mathrm{pr},i}(q_{\mathrm{pr},i})=-\theta_i \vert q_{\mathrm{pr},i} \vert  q_{\mathrm{pr},i},~~i=1,...,n_\mathrm{pr}.
\end{equation}
This fact is fundamental to our developments in Section~\ref{sec:3}. We point out that the explicit form of $f_{\mathrm{ch},i}$ is not necessary to establish the main results and conclusions of this work, nonetheless, some additional details appear in Remark~\ref{rem:add_details_fch} below.\footnote{The validity of \eqref{eq:flow_ODE} relies on two additional assumptions concerning the topology of the distribution network and the placement of the system's pumps. Including these assumptions here would require a considerable amount of additional background information that we omit due to space constraints. All the details can be found in \cite{jmc_dh_modeling_2020}.}
}
{\color{black}
\begin{remark}
Strictly speaking $f_{\mathrm{pr},i}$ represents the direct sum of the pressure drops caused by the series connection of a given producer's pipe and valve (see Fig.~\ref{fig:2}).  Thus,   $f_{\mathrm{pr},i}(q_{\mathrm{pr},i})=-\left(\sum_{j\in \mathcal{P}_i}\theta_j\right) \vert q_{\mathrm{pr},i} \vert  q_{\mathrm{pr},i}$, where $\mathcal{P}_i\subset \mathcal{E}$ comprises the valve and the pipe associated to the $i$th producer. Out of simplicity we use the coefficient $\theta_i$ in \eqref{eq:f_pr}, however this does not affect our conclusions since both $\theta_i$ and $\left(\sum_{j\in \mathcal{P}_i}\theta_j\right)$ are unknown positive constants for any $i=1,...,n_\mathrm{pr}$. 
\end{remark}
}

\bigskip

\noindent {\color{black}We move on to the nodes' volume dynamics of \eqref{eq:vol_dyn_NR}. By recalling that the DH system's graph is connected, it follows that $1_{n_\mathrm{N}}^\top \dot{V}_\mathrm{N}=1_{n_\mathrm{N}}^\top \mathcal{B}_0 q_\mathrm{E}=0$.  Let us respectively denote by $V_{\mathrm{sh},i}$ and $V_{\mathrm{sc},i}$ the volume of water in the hot and cold layers of the storage tank $i$, and recall that we have assumed $V_{\mathrm{N},k}=0$ (constant) for each simple junction $k\in\mathcal{N}$. Then, $1_{n_\mathrm{N}}^\top \dot{V}_\mathrm{N}=0$ implies $\sum_{i=1}^{n_\mathrm{ST}}\dot{V}_{\mathrm{sh},i}+\dot{V}_{\mathrm{sc},i}=0$ for all time. Thus,  volume is preserved in the DH system and, notably, each  storage tank is completely filled with water at all times, provided that the same condition is met at $t=0$, in satisfaction of assumption {\bf (b)}. The latter assumption is commonly found in related literature and is usually understood as the desired operation mode of this type of devices (see, {e.g.}, \cite{Scholten_tcst_2015,verda_storage_11,kamal_storage_97}).} {\color{black}It follows  that by considering  \eqref{eq:q_E} and  assumptions {\bf (a)} and {\bf (c)}, then the model \eqref{eq:vol_dyn_NR} can be reduced to the following ODE:
\begin{equation}\label{eq:Vsh_dyn}
\dot{V}_\mathrm{sh}=q_\mathrm{pr}-Bq_\mathrm{ch},
\end{equation}
where  $V_\mathrm{sh}\in \mathbb{R}^{n_\mathrm{ST}}$, with $n_\mathrm{ST}=n_\mathrm{pr}$, comprises the volumes of water in  the hot layers of the storage tanks, and   $B$ is an appropriate sub-block of matrix $\mathcal{F}$  in \eqref{eq:q_E} (thus $B_{ij}\in \{-1,0,1 \}$). We remark that $(Bq_\mathrm{ch})_i$ represents  the flow at the hot layer's outlet of the storage tank to which we associate $V_{\mathrm{sh},i}$ and  that, necessarily,  $\dot{V}_\mathrm{sc}=Bq_\mathrm{ch}-q_\mathrm{pr}.$
}

\begin{figure}[t]
\begin{center}
\includegraphics[width=0.9\linewidth]{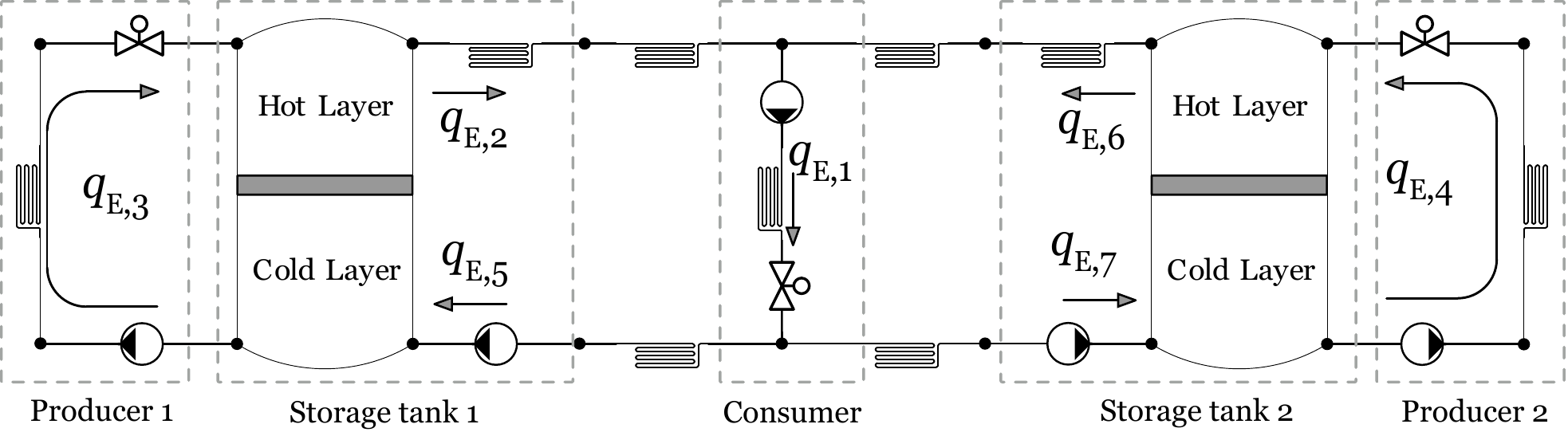}
\caption{Schematic diagram of a simplified DH system with two producers and one consumer. Each producer is interfaced to the distribution network through a storage tank. Only the flows through selected devices are shown.}
\label{fig:3}
\end{center}
\end{figure}

{\color{black}
Before continuing to this paper's next section, a number of remarks are in order:

}

{\color{black}
\begin{remark}\label{rem:add_details_fch}
The matrix $\mathcal{F}$ introduced in \eqref{eq:fundamental_loop} is related to the {\em fundamental loop matrix} of the DH system's graph. Then, equation \eqref{eq:q_E} corresponds to the  hydraulic analogous of Kirchhoff's current law.  Also, the analogy between pressures and voltages allows us to invoke Kirchhoff's voltage law analogous for hydraulic networks, which states  that the sum of the pressure drops  across any fundamental loop of the DH system's graph is zero, or equivalently, that
\begin{equation}\label{eq:kirchh_volt}
\mathcal{F}\Delta P_\mathrm{E}=0,
\end{equation}
where each component  $\Delta P_{\mathrm{E},i}$ corresponds to the pressure drop across any edge $i\in \mathcal{E}$, then $\Delta P_{\mathrm{E},i}$ is equivalent to \eqref{eq:dae_flows}.  It follows that through an adequate (and consistent) ordering of the components of $q_\mathrm{E}$, $f_\mathrm{E}$ and $\mathcal{F}$,  the decoupled, ODE-based dynamics \eqref{eq:flow_ODE} can be obtained from \eqref{eq:dae_flows} via the algebraic constraints \eqref{eq:q_E} and \eqref{eq:kirchh_volt}. Moreover, the  mappings $f_\mathrm{ch}$ and $f_\mathrm{pr}$ are given by
\begin{equation*}
\begin{bmatrix}
f_\mathrm{ch}(q_\mathrm{ch})\\
f_\mathrm{pr}(q_\mathrm{pr})
\end{bmatrix}=-\mathcal{F}f_\mathrm{E}(q_\mathrm{E})\vert_{q_\mathrm{E}=\mathcal{F}^\top\begin{bmatrix}
q_\mathrm{ch}\\
q_\mathrm{pr}
\end{bmatrix}}.
\end{equation*}
The mappings  $f_\mathrm{ch}$ and $f_\mathrm{pr}$ are continuously differentiable and, since  $\mathcal{F}$ is a full row rank matrix and each $f_{\mathrm{E},i}$ associated to a pipe (and a valve) is a monotone function, then $-f_\mathrm{ch}$ and $-f_\mathrm{pr}$ are monotone mappings \cite{jmc_dh_modeling_2020}. Out of completeness we indicate that, for each $i=1,...,n_\mathrm{ch}$, we have that
\begin{align*}
f_{\mathrm{ch},i}(q_\mathrm{ch}) & =-\theta_i \vert q_{\mathrm{ch},i} \vert q_{\mathrm{ch},i} \nonumber\\
&~~ -\sum_{j=1}^{n_\mathrm{g}} G_{ij} \theta_{\small j+n_\mathrm{pr}+n_\mathrm{ch}} \sum_{k=1}^{n_\mathrm{ch}} \vert G_{ki}q_{\mathrm{ch},k} \vert  G_{ki}q_{\mathrm{ch},k},
\end{align*}
with	 $n_\mathrm{g}$ denoting the number of columns of matrix $G$ in \eqref{eq:fundamental_loop}. Also, for each $i=1,...,n_\mathrm{pr}$,
\begin{align*}
f_{\mathrm{pr},i}(q_\mathrm{pr}) & =-\theta_{\small i+n_\mathrm{ch}} \vert q_{\mathrm{pr},i} \vert q_{\mathrm{pr},i} \nonumber\\
&~~ -\sum_{j=1}^{n_\mathrm{h}} H_{ij} \theta_{\small j+n_\mathrm{ch}+n_\mathrm{pr}+n_\mathrm{g}} \sum_{k=1}^{n_\mathrm{pr}} \vert H_{ki}q_{\mathrm{pr},k} \vert  H_{ki}q_{\mathrm{pr},k},
\end{align*}
where $n_\mathrm{h}$ is the number of columns of matrix $H$ in \eqref{eq:fundamental_loop}.  Due to the considered topology to interface each producer to the distribution network (see Fig.~\ref{fig:3}), the above expression for $f_{\mathrm{pr},i}$ can be simplified to the expression in \eqref{eq:f_pr}. See \cite[Section~2.1]{jmc_dh_modeling_2020} for details.
\end{remark}
}

{\color{black}

\begin{remark}
Since we consider that  producers and consumers have same topology (see Fig.~\ref{fig:2}), then it is possible to interface any number of producers directly to the distribution network without a storage tank. It can be shown that in such circumstances   the analysis and main conclusions of this work would not change if assumption {\bf (a)} is discarded. We emphasize nonetheless that the combination of assumptions {\bf (a)}, {\bf (b)} and {\bf (c)} allows us to decouple  the dynamics between $q_\mathrm{pr}$ and $q_\mathrm{ch}$, to write each mapping $f_{\mathrm{pr},i}$ in \eqref{eq:f_pr} only in terms of  $q_{\mathrm{pr},i}$, and  making  $q_\mathrm{pr}$ in the right-hand side of \eqref{eq:Vsh_dyn}  appear with an identity coefficient matrix. The relaxation of assumption {\bf (c)} is part of our ongoing research.
\end{remark}

}

\subsection{Problem formulation}

This letter is concerned with the objective of simultaneously regulating the  DH system's  flow and volume vectors $q_\mathrm{ch}$ and $V_\mathrm{sh}$ towards desired constant setpoints $q_\mathrm{ch}^\star$ and $V_\mathrm{sh}^\star$, respectively. More precisely, our goal is that the overall flow and volume dynamics conformed by \eqref{eq:flow_ODE} and \eqref{eq:Vsh_dyn} attain
\begin{align*}
\lim_{t\rightarrow \infty}q_\mathrm{ch}=q^\star_\mathrm{ch},~\mathrm{and}~\lim_{t\rightarrow \infty}V_\mathrm{sh}=V^\star_\mathrm{sh},
\end{align*}
for an identifiable  set of initial conditions. 

Concerning the regulation of $q_\mathrm{ch}$,  we note that the heat transport from producers to consumers  depends strongly on the flows through the distribution network; in particular,  consumers usually regulate their heat demand and temperature indirectly by adjusting their flow. The regulation  of $V_\mathrm{sh}$ on the other hand is relevant for augmenting or reducing the stored (useful) energy in a tank, as the latter is proportional to the volume in the hot layer of the tank \cite{Scholten_tcst_2015}. Therefore, the achievement	 of the  above-mentioned objective is fundamental  for the correct transport and management of the DH system's energy resources.

{\color{black}To achieve the desired objective, we  design for each of the system inputs $u_{\mathrm{ch},i}$ and $u_{\mathrm{pr},i}$, decentralized and dynamic control laws of the form
\begin{equation}\label{eq:dynamic_control_generic}
\begin{aligned}
\dot{x}_{\mathrm{c},i} & = g_{\mathrm{c},i}\left(\xi_{\mathrm{c},i}, x_{\mathrm{c},i} \right),\\
u_{\mathrm{c},i} & = h_{\mathrm{c},i}\left(\xi_{\mathrm{c},i}, x_{\mathrm{c},i} \right),
\end{aligned}
\end{equation}
	where $\mathrm{c} \in \{\mathrm{ch}, \mathrm{pr}\}$, $x_{\mathrm{c},i}$ is the state of the controller and $\xi_{\mathrm{c},i}$ is a vector comprising signals and parameters available to $u_{\mathrm{c},i}$. By decentralized we mean that $u_{\mathrm{c},i}$ uses information that is locally available at its associated heat producer, consumer or distribution pipe. More precisely, we assume  that 
\begin{equation}\label{eq:known_data}
\begin{aligned}
\xi_{\mathrm{ch},i} & = \begin{bmatrix} q_{\mathrm{ch},i}, & q_{\mathrm{ch},i}^\star \end{bmatrix}^\top,\\
\xi_{\mathrm{pr},i} & = \begin{bmatrix}
q_{\mathrm{pr},i}, &  V_{\mathrm{sh},i}, & \left(Bq_\mathrm{ch} \right)_i, & J_{\mathrm{pr},i}, & V_{\mathrm{sh},i}^\star
\end{bmatrix}^\top.
\end{aligned}
\end{equation}
\begin{remark}\label{rem:known_data}
We recall  that $(Bq_\mathrm{ch})_i$ represents the flow at the $i$th tank's hot layer outlet (see Fig.~\ref{fig:3}) and we assume that it can be measured locally by the $i$th producer. Then, we underscore  that the knowledge of $B$  is not needed to establish our main results in Section~\ref{sec:3} and it is not needed to compute $\xi_{\mathrm{ch},i}$ nor $\xi_{\mathrm{pr},i}$. The same applies to each $\theta_i$, as it has been assumed to be an unknown positive scalar (see \eqref{eq:model_f_E} and \eqref{eq:f_pr}).
\end{remark}
}

\section{Volume regulation controller}\label{sec:3}

In this section we detail the aspects of the proposed solution to the formulated problem.  In view of the  cascade structure of the open-loop  flow and volume dynamics  \eqref{eq:flow_ODE} and \eqref{eq:Vsh_dyn}, we first present a decentralized, proportional-integral controller for the stabilization of the subsystem \eqref{eq:flow_ODE_qch}. The latter controller is inspired by the results of \cite{DePersis2014} (see also \cite[Sec. 3.1]{jmc_dh_modeling_2020}) {\color{black} addressing end-user {\em pressure} regulation in {\em single}-producer DH systems {\em without} storage units}.  Afterwards, we focus on the remaining dynamics, i.e., in \eqref{eq:flow_ODE_qpr} and \eqref{eq:Vsh_dyn},  and propose a novel adaptive control scheme to achieve volume regulation of the storage tanks.

\begin{proposition}\label{prop:stab_q_ch}
Consider the following controller of the form introduced in \eqref{eq:dynamic_control_generic}:
\begin{equation}\label{eq:control_q_ch_1}
\begin{aligned}
\dot{x}_\mathrm{ch} &  = -M_\mathrm{ch}\left(q_\mathrm{ch}-q^\star_\mathrm{ch} \right)\\
u_\mathrm{ch} & = -N_\mathrm{ch}\left(q_\mathrm{ch}-q^\star_\mathrm{ch} \right)+x_\mathrm{ch},
\end{aligned}
\end{equation}
where $M_\mathrm{ch}$ and $N_\mathrm{ch}$ are constant, positive definite diagonal matrices. Then, the subsystem \eqref{eq:flow_ODE_qch} in closed-loop with \eqref{eq:control_q_ch_1}   admits a globally  asymptotically stable equilibrium point. Moreover, $\lim_{t\rightarrow \infty}q_\mathrm{ch}=q^\star_\mathrm{ch}$.
\end{proposition}

\begin{proof}
{\color{black}The  closed-loop system \eqref{eq:flow_ODE_qch}, \eqref{eq:control_q_ch_1}  is given by 
\begin{equation}\label{eq:chords_dyn_CL}
\begin{aligned}
J_\mathrm{ch}\dot{q}_\mathrm{ch} & = f_\mathrm{ch}(q_\mathrm{ch})-N_\mathrm{ch}(q_\mathrm{ch}-q_\mathrm{ch}^\star) +x_\mathrm{ch}\\
\dot{x}_\mathrm{ch} & = -M_\mathrm{ch}(q_\mathrm{ch}-q_\mathrm{ch}^\star).
\end{aligned}
\end{equation}
The right-hand side vector field of this system is continuously differentiable, as $f_\mathrm{ch}$ also is. Moreover, $(\bar{q}_\mathrm{ch},\bar{x}_\mathrm{ch})=(q_\mathrm{ch}^\star, -f_\mathrm{ch}(q_\mathrm{ch}^\star))$ is the unique equilibrium of \eqref{eq:chords_dyn_CL}.} Let us define the following function:
\begin{align}\label{eq:stor_ch}
\mathcal{S}_\mathrm{ch}(q_\mathrm{ch},x_\mathrm{ch})& =\tfrac{1}{2}(q_\mathrm{ch}-\bar{q}_\mathrm{ch})^\top J_\mathrm{ch}(q_\mathrm{ch}-\bar{q}_\mathrm{ch}) \nonumber\\
& +\tfrac{1}{2}\left(x_\mathrm{ch}-\bar{x}_\mathrm{ch}\right)^\top M_\mathrm{ch}^{-1}\left(x_\mathrm{ch}-\bar{x}_\mathrm{ch}\right),
\end{align}
which is positive definite {\color{black}and radially unbounded}.  The time derivative of $\mathcal{S}_\mathrm{ch}$ along the trajectories of \eqref{eq:chords_dyn_CL} satisfies
{\color{black}
\begin{align*}
 \dot{\mathcal{S}}_\mathrm{ch} &  = (q_\mathrm{ch}-\bar q_\mathrm{ch})^\top J_\mathrm{ch}\dot{q}_\mathrm{ch}+ (x_\mathrm{ch}-\bar{x}_\mathrm{ch})^\top M_\mathrm{ch}^{-1} \dot{x}_\mathrm{ch}\\
 & =  - \left(q_\mathrm{ch}-\bar{q}_\mathrm{ch}\right)^\top N_\mathrm{ch}(q_\mathrm{ch}-\bar q_\mathrm{ch}) \\
 &~~~ +(q_\mathrm{ch}-\bar{q}_\mathrm{ch})^\top \left( f_\mathrm{ch}(q_\mathrm{ch})-f_\mathrm{ch}(\bar{q}_\mathrm{ch})\right)\\
 & \leq -(q_\mathrm{ch}-\bar q_\mathrm{ch})^\top N_\mathrm{ch} (q_\mathrm{ch}-\bar q_\mathrm{ch}),
\end{align*}
where to obtain the latter inequality we have used the equilibrium identity $\bar{x}_\mathrm{ch} = -f_\mathrm{ch}(\bar q_\mathrm{ch})$ together with the fact that $(q_\mathrm{ch}-\bar{q}_\mathrm{ch})^\top \left( f_\mathrm{ch}(q_\mathrm{ch})-f_\mathrm{ch}(\bar{q}_\mathrm{ch})\right)\leq 0$, which holds by virtue of  $-f_\mathrm{ch}$ being a monotone mapping (see \cite[Lemma~4]{jmc_dh_modeling_2020}).} {\color{black}We note that $\mathcal{S}_\mathrm{ch}$ is not a Lyapunov function  since its derivative is not (strictly) negative definite with respect to $(\bar{q}_\mathrm{ch},\bar{z}_\mathrm{ch})$. Nonetheless, the only solution of \eqref{eq:chords_dyn_CL} that can stay identically in the set $E=\{(q_\mathrm{ch},x_\mathrm{ch}):~\dot{\mathcal{S}}_\mathrm{ch}=0 \Leftrightarrow q_\mathrm{ch}=\bar{q}_\mathrm{ch}\}$ is the equilibrium $(\bar{q}_\mathrm{ch},\bar{z}_\mathrm{ch})$. Following LaSalle's invariance principle, it is thus concluded that $(\bar{q}_\mathrm{ch},\bar{z}_\mathrm{ch})$ is globally asymptotically stable (see \cite[Theorem~3.5]{khalil_book_NLCONTROL_GLOBAL}).}
\end{proof}

{\color{black}
\begin{remark}
In the proof of Proposition~\ref{prop:stab_q_ch} we have used the fact that $-f_\mathrm{ch}$ is a monotone mapping, however it was not necessary to provide an explicit expression for it. Then, other, possibly more general, models to describe the pressure  drop $f_{\mathrm{E},i}$ in pipes and valves could be used instead of \eqref{eq:model_f_E} in Assumption~\ref{assu:f_E_pipe} (as long as $f_{\mathrm{E},i}$ is a monotone function). Moreover, the monotonicity of $-f_\mathrm{ch}$ also implies that the open-loop system \eqref{eq:flow_ODE_qch} is {\em shifted passive} \cite{nima_scl_19} with  storage function  
\begin{equation*}
\mathcal{H}(q_\mathrm{ch})=\tfrac{1}{2}(q_\mathrm{ch}-\bar{q}_\mathrm{ch})^\top J_\mathrm{ch}(q_\mathrm{ch}-\bar{q}_\mathrm{ch})
\end{equation*}
and passive output $q_{\mathrm{ch}}$. Therefore, along any solution $q_\mathrm{ch}$ of  \eqref{eq:flow_ODE_qch}, the following inequality is satisfied
\begin{equation*}
\dot{\mathcal{H}}(q_\mathrm{ch}) \leq (u_\mathrm{ch}-\bar{u}_\mathrm{ch})^\top (q_\mathrm{ch}-\bar{q}_\mathrm{ch}),
\end{equation*}
for any equilibrium pair $(\bar{u}_\mathrm{ch}, \bar{q}_\mathrm{ch})$. Based on the results of \cite{bayu_scl_07} for the stabilization of nonlinear RLC circuits and of \cite{DePersis2014} for pressure regulation of single-producer DH systems, this observation is the main motivation to propose the simple proportional-integral controller \eqref{eq:control_q_ch_1} to achieve the desired objective of regulating $q_\mathrm{ch}$ towards a constant setpoint. 
\end{remark}
}

\bigskip

\noindent {\color{black}Now, we turn our attention to the problem of  regulating  the volume of hot water of each storage tank towards constant, specified setpoints.  Then, we focus on  the system \eqref{eq:flow_ODE_qpr},  \eqref{eq:Vsh_dyn}, which we write next in a more suitable equivalent form. On the one hand, considering  \eqref{eq:f_pr} (see also Assumption~\ref{assu:f_E_pipe}), the mapping $f_\mathrm{pr}$ in \eqref{eq:flow_ODE_qpr} can be written as
\begin{equation}\label{eq:explicit_fpr}
f_\mathrm{pr}(q_\mathrm{pr})=- W(q_\mathrm{pr}) \theta,
\end{equation}
where $\theta  = \begin{bmatrix}\theta_1,~\theta_2,~\cdots~ \theta_{n_\mathrm{pr}} \end{bmatrix}^\top$  and  $W(q_\mathrm{pr})  = \langle \vert q_{\mathrm{pr},i}\vert  q_{\mathrm{pr},i}  \rangle_{i=1}^{n_\mathrm{pr}}$.  On the other hand, \eqref{eq:Vsh_dyn} can be equivalently written as
\begin{equation}\label{eq:vsh_new}
\dot{V}_\mathrm{sh}=q_\mathrm{pr}-Bq_\mathrm{ch}^\star+\Psi(q_\mathrm{ch}),
\end{equation}
where $B q_\mathrm{ch}^\star$  and  
\begin{equation}\label{eq:Psi}
\Psi(q_\mathrm{ch})=B(q_\mathrm{ch}^\star-q_\mathrm{ch}).
\end{equation}
In view of \eqref{eq:explicit_fpr}-\eqref{eq:Psi},   the system of interest to address storage volume regulation is equivalent to
\begin{subequations}\label{eq:equiv_qpr_vsh}
\begin{align}
J_\mathrm{pr}\dot{q}_\mathrm{pr} & =  -W(q_\mathrm{pr})\theta+ u_\mathrm{pr}\\
\dot{V}_\mathrm{sh} & = q_\mathrm{pr}-Bq_\mathrm{ch}^\star + \Psi(q_\mathrm{ch}),
\end{align}
\end{subequations}
where $W(q_\mathrm{pr})\theta$, $Bq_\mathrm{ch}^\star$ and $\Psi(q_\mathrm{ch})$ act as disturbances. Indeed, we have considered that both $\theta$ and $B$ are unknown to the DH system's input vectors $u_\mathrm{ch}$ and $u_\mathrm{pr}$ (see Remark~\ref{rem:known_data}). Also, the vector  $q_{\mathrm{ch}}^\star$ is not necessarily available to  each $u_{\mathrm{pr},i}$ as there is no communication among producers and consumers.\footnote{Henceforth we assume that  \eqref{eq:flow_ODE_qch} is in closed-loop with \eqref{eq:control_q_ch_1}, then $\Psi(q_\mathrm{ch})$ is a bounded and  vanishing disturbance to \eqref{eq:equiv_qpr_vsh}.}}

In the next proposition, {\color{black}by provisionally neglecting}  the effect of the disturbance  $\Psi(q_\mathrm{ch})$, we present a stabilizing controller for \eqref{eq:equiv_qpr_vsh}.  The proposed, suitably-tailored dynamic controller for $u_\mathrm{pr}$ attains  asymptotic convergence of $V_\mathrm{sh}$ towards a desired constant value and estimates in real-time the unknown parameter vector $\theta$. {\color{black}This result will be fundamental in Theorem~\ref{theorem:objective_1} where, using cascade system arguments, we establish the asymptotic stability of the overall DH system's closed-loop dynamics, with $\Psi(q_\mathrm{ch})$ now acting on \eqref{eq:equiv_qpr_vsh}.}

\begin{proposition}\label{prop:stab_vsh}
Consider the system \eqref{eq:equiv_qpr_vsh} and assume that $\Psi(q_\mathrm{ch})=0$ for all time. Define
\begin{equation}\label{eq:qpr_to_zpr}
z_\mathrm{pr}  := q_\mathrm{pr}-x_\mathrm{a}+N_\mathrm{sh}\left(V_\mathrm{sh}-V_\mathrm{sh}^\star \right)
\end{equation}
and consider the following dynamic controller
\begin{subequations}\label{eq:controller_qpr_vsh}
\begin{align}
\dot{x}_\mathrm{a} & = -M_\mathrm{a}\left(V_\mathrm{sh}-V_\mathrm{sh}^\star \right) \label{eq:dot_xa}\\
\dot{x}_\mathrm{b} & = -M_\mathrm{b} \tilde{W}(z_\mathrm{pr}) z_\mathrm{pr} \label{eq:dot_xb}\\
u_\mathrm{pr} & =\tilde{W}(z_\mathrm{pr})x_\mathrm{b}- \left( J_\mathrm{pr}\left(M_\mathrm{a}-N_\mathrm{sh}^2 \right)+I \right)(V_\mathrm{sh}-V_\mathrm{sh}^\star) \nonumber \\
& ~~~ - \left( J_\mathrm{pr}N_\mathrm{sh}+N_\mathrm{pr} \right)z_\mathrm{pr}+J_\mathrm{pr}N_\mathrm{sh}\left(Bq_\mathrm{ch}-x_\mathrm{a}\right), \label{eq:upr_law}
\end{align}
\end{subequations}
where $N_\mathrm{pr}$, $N_\mathrm{sh}$, $M_\mathrm{a}$ and $M_\mathrm{b}$ are constant, positive definite diagonal matrices, and $\tilde{W}(z_\mathrm{pr}):=W(q_\mathrm{pr})\vert_{q_\mathrm{pr}=z_\mathrm{pr}+x_\mathrm{a}-N_\mathrm{sh}(V_\mathrm{sh}-V_\mathrm{sh}^\star)}$. If the closed-loop system admits an equilibrium such that $\tilde{W}(\bar{z}_\mathrm{pr})$ is not identically zero, then said equilibrium is globally asymptotically stable. Moreover,  $\lim_{t\rightarrow \infty}V_\mathrm{sh}=V_\mathrm{sh}^\star$, where $V_\mathrm{sh}^\star$ is a predefined, constant setpoint, and $\lim_{t\rightarrow \infty}x_\mathrm{b}=\theta$.
\end{proposition}
\begin{proof}
Inspired by backstepping control design \cite[Chapter~9]{khalil_book_NLCONTROL_GLOBAL}, we propose first a change of variable from $q_\mathrm{pr}$ to $z_\mathrm{pr}$ as appears in \eqref{eq:qpr_to_zpr}. Then, \eqref{eq:equiv_qpr_vsh}, with $\Psi(q_\mathrm{ch})=0$,  is transformed into the equivalent system
\begin{equation}\label{eq:aux_almostCL}
\begin{aligned}
J_\mathrm{pr}\dot{z}_\mathrm{pr} & = -\tilde{W}(z_\mathrm{pr})\theta+ J_\mathrm{pr}\left(M_\mathrm{a}-N_\mathrm{sh}^2\right)\left(V_\mathrm{sh}-V_\mathrm{sh}^\star\right)\\
& ~~~ +J_\mathrm{pr}N_\mathrm{sh}\left(z_\mathrm{pr}+x_\mathrm{a}-Bq_\mathrm{ch}^\star \right)+ u_\mathrm{pr}\\
\dot{V}_\mathrm{sh} & = z_\mathrm{pr}+x_\mathrm{a}-N_\mathrm{sh}\left(V_\mathrm{sh}-V^\star_\mathrm{sh}\right) -B q_\mathrm{ch}^\star \\
\dot{x}_\mathrm{a} & = -M_\mathrm{a}\left( V_\mathrm{sh}-V_\mathrm{sh}^\star \right).
\end{aligned}
\end{equation}
Substituting \eqref{eq:upr_law} into \eqref{eq:aux_almostCL}, which in addition brings the variable $x_\mathrm{b}$ satisfying \eqref{eq:dot_xb}, illustrates  the  IDA-PBC \cite{ortega_garcia_04} feature of the controller by virtue of  attaining a closed-loop system that we write in Hamiltonian form as follows:
{\color{black}
\begin{equation}\label{eq:CL_by_comps_VSH}
\begin{aligned}
J_\mathrm{pr}\dot{z}_\mathrm{pr} & = -N_\mathrm{pr}z_\mathrm{pr}-(V_\mathrm{sh}-V_\mathrm{sh}^\star)+\tilde{W}(z_\mathrm{pr})(x_\mathrm{b}-\theta )\\
\dot{V}_\mathrm{sh} & = z_\mathrm{pr}-N_\mathrm{sh}(V_\mathrm{sh}-V_\mathrm{sh}^\star)+(x_\mathrm{a}-B q_\mathrm{ch}^\star)\\
M_\mathrm{a}^{-1}\dot{x}_\mathrm{a} & = -(V_\mathrm{sh}-V_\mathrm{sh}^\star)\\
M_\mathrm{b}^{-1}\dot{x}_\mathrm{b} & = -\tilde{W}(z_\mathrm{pr})z_\mathrm{pr},
\end{aligned}
\end{equation}
}
which is equivalent to:
\begin{equation}\label{eq:Sigma_prsh_prime}
\dot{X}=\underbrace{\begin{bmatrix}
-N_\mathrm{pr} & -I & 0 & \tilde{W}(z_\mathrm{pr})\\
I & -N_\mathrm{sh} & I & 0\\
0 & -I & 0 & 0\\
-\tilde{W}(z_\mathrm{pr}) & 0 & 0 & 0
\end{bmatrix}}_{=:{F}(X)} \nabla \tilde{H}(X),
\end{equation}
with state vector $X=\left(J_\mathrm{pr}z_\mathrm{pr},V_\mathrm{sh}, M_\mathrm{a}^{-1} x_\mathrm{a}, M_\mathrm{b}^{-1} x_\mathrm{b}  \right)$ and Hamiltonian
\begin{align*}
\tilde{H}= \frac{1}{2}\left(X-\bar{X} \right)^\top\text{block.diag}\left(J_\mathrm{pr}^{-1},I, M_\mathrm{a}, M_\mathrm{b} \right) \left( X-\bar{X} \right),
\end{align*}
where 
\begin{equation}\label{eq:bar X}
\bar X= (0,V^\star_\mathrm{sh}, M_\mathrm{a}^{-1}B q^\star_\mathrm{ch},M_\mathrm{b}^{-1} \theta)
\end{equation}
is a constant vector.
{\color{black} Considering the definition of the diagonal matrix $\tilde{W}(z_\mathrm{pr})$, it is straightforward to see that the right-hand side vector field of the ODE \eqref{eq:Sigma_prsh_prime} is  continuously differentiable. Also,  we  have assumed that $\tilde{W}(\bar{z}_\mathrm{pr})$ is not identically zero, then ${F}(\bar{X})$ in \eqref{eq:Sigma_prsh_prime} has full rank, implying that $X=\bar X$ is  the unique equilibrium point of \eqref{eq:Sigma_prsh_prime}.} We  underscore  that the {\em real-time} estimation of $\theta$, which  is obtained  from $x_\mathrm{b}$, represents  the adaptive aspect of the proposed controller ({c.f.}, \cite[Example~4]{nageshrao_ph_adaptive_16}),  provided that \eqref{eq:Sigma_prsh_prime} is asymptotically stable. Next  we show, using  LaSalle's invariance principle, that $\bar X$ {is} globally asymptotically stable. {\color{black}Consider the Hamiltonian $\tilde{H}$ and observe that it is positive definite with respect to $\bar X$ and radially unbounded with respect to $X$.} Moreover, along the solutions of \eqref{eq:Sigma_prsh_prime} we have that:
{\color{black}
\begin{align*}
 \dot{\tilde{H}}  & =  \left(\nabla \tilde{H}(X)\right)^\top \dot{X}\\
& = \frac{1}{2}\left(\nabla \tilde{H}(X)\right)^\top \left({F}(X)+\mathcal{F}^\top (X) \right)\nabla \tilde{H}(X) \\
& = -z_\mathrm{pr}^\top N_\mathrm{pr} z_\mathrm{pr}-(V_\mathrm{sh}-V^\star_\mathrm{sh})^\top N_\mathrm{sh}(V_\mathrm{sh}-V^\star_\mathrm{sh}) \leq 0.
\end{align*}
We see that   $\dot{\tilde{H}}$ may be zero at values of $X\neq \bar X$, implying that $\tilde{H}$ is not a Lyapunov function. However, we note from \eqref{eq:CL_by_comps_VSH}  that no solution of this system can stay in  $S=\{X:~\dot{\tilde H}(X)=0\Leftrightarrow z_\mathrm{pr}=0,~V_\mathrm{sh}=V_\mathrm{sh}^\star \}$ other than $\bar{X}$. {\color{black} Indeed, let $X_S$ be an arbitrary solution of \eqref{eq:CL_by_comps_VSH} that remains in $S$ for all $t\geq 0$, then $z_\mathrm{pr}=0$ and $V_\mathrm{sh}=V_\mathrm{sh}^\star$ for all $t\geq 0$. Since the diagonal matrix $\tilde{W}(\bar{z}_\mathrm{pr}=0)$ is different from zero (by assumption) and hence, non singular, it follows that $X_S$ can identically stay in $S$ if and only if $X_S=\bar{X}$ for all $t\geq 0$.} It is concluded then, invoking LaSalle's invariance principle (see \cite[Theorem~3.5]{khalil_book_NLCONTROL_GLOBAL}),  that $\bar X$  is {\em globally} asymptotically stable.}
\end{proof}

Before  presenting the next result, which concerns the asymptotic stability of the overall DH system hydraulic dynamics,  consider the following:

\begin{remark}\label{rem:PE}
{\color{black}The  assumption about the equilibrium of the closed-loop system  \eqref{eq:equiv_qpr_vsh}, \eqref{eq:qpr_to_zpr}, \eqref{eq:controller_qpr_vsh}  satisfying $\tilde{W}(\bar{z}_\mathrm{pr})\neq 0$, which in view of  $\tilde{W}(z_\mathrm{pr}) = W(q_\mathrm{pr})$ is equivalent to $\bar q_{\mathrm{pr},i}\neq0$ (for all $i$), notably guarantees that $x_\mathrm{b}\rightarrow \theta$ asymptotically, {i.e.}, the vector of unknown coefficients $\theta$ can be accurately estimated via the proposed adaptive scheme ({c.f.}, \cite[Example~4]{nageshrao_ph_adaptive_16}), overcoming the  challenging unknown and time-varying disturbance $-W(q_\mathrm{pr})\theta$ acting on \eqref{eq:equiv_qpr_vsh}, which we recall  stems from Assumption~\ref{assu:f_E_pipe}.} {\color{black} Considering \eqref{eq:Vsh_dyn}, it is clear that a necessary and sufficient condition for $\bar{q}_\mathrm{pr}\neq 0$ is that $B\bar{q}_\mathrm{ch}\neq 0$, which is a condition that can potentially be enforced through an adequate choice of the setpoint $q_\mathrm{ch}^\star$ (see Proposition~\ref{prop:stab_q_ch}).} We note also that a steady-state condition in which $\bar q_{\mathrm{pr},i}=0$, for some index $i$, implies that the associated heat producer is not in operation.
\end{remark}

\begin{theorem}\label{theorem:objective_1}
The overall closed-loop flow and volume dynamics of the DH system, described by \eqref{eq:flow_ODE},  \eqref{eq:control_q_ch_1}, \eqref{eq:qpr_to_zpr} and \eqref{eq:controller_qpr_vsh}  has a (locally) asymptotically stable equilibrium, and
 \begin{align}
\lim_{t\rightarrow \infty}q_\mathrm{ch}=q^\star_\mathrm{ch},~\mathrm{and}~\lim_{t\rightarrow \infty}V_\mathrm{sh}=V^\star_\mathrm{sh},
\end{align}
where $q^\star_\mathrm{ch}$ and $V^\star_\mathrm{sh}$ are pre-specified setpoints, provided that $\bar q_{\mathrm{pr},i} \neq 0$, {\color{black}for all} $i=1,2,...n_\mathrm{pr}$.  
\end{theorem}
\begin{proof}
{\color{black}
Let us identify by $\Sigma_\mathrm{ch}$ the system conformed by \eqref{eq:flow_ODE} and  \eqref{eq:control_q_ch_1}, which we write from \eqref{eq:chords_dyn_CL} as follows:
\begin{equation}\label{eq:Sigma_ch}
\Sigma_\mathrm{ch}:\left\{
\begin{tabular}{rcl}
$J_\mathrm{ch}\dot{q}_\mathrm{ch}$ & = & $f_\mathrm{ch}(q_\mathrm{ch})-N_\mathrm{ch}(q_\mathrm{ch}-q_\mathrm{ch}^\star) +x_\mathrm{ch}$\\
$\dot{x}_\mathrm{ch}$ & = & $-M_\mathrm{ch}(q_\mathrm{ch}-q_\mathrm{ch}^\star).$
\end{tabular}
 \right.
\end{equation}
Also, let $\Sigma_\mathrm{pr,sh}$ denote the dynamics of  \eqref{eq:flow_ODE_qpr}, \eqref{eq:Vsh_dyn}, which is  equivalent to \eqref{eq:equiv_qpr_vsh}, in closed-loop with \eqref{eq:controller_qpr_vsh}. Considering the change of variable (from $q_\mathrm{pr}$ to $z_\mathrm{pr}$) of equation \eqref{eq:qpr_to_zpr}, together with \eqref{eq:Sigma_prsh_prime}, this system is equivalent to
\begin{equation}\label{eq:Sigma_prsh}
\Sigma_\mathrm{pr,sh}:\left\{
\dot{X}={F}(X)\nabla \tilde{H}(X)+\begin{bmatrix}
0_{n_\mathrm{pr}}\\
\Psi(q_\mathrm{ch})\\
0_{n_\mathrm{pr}}\\
0_{n_\mathrm{pr}}
\end{bmatrix},
 \right.
\end{equation}
where $X$, ${F}(X)$ and $\tilde{H}$ are the same as for \eqref{eq:Sigma_prsh_prime}. We note that, compared to \eqref{eq:Sigma_prsh_prime}, the system $\Sigma_\mathrm{pr,sh}$ includes the unknown and bounded disturbance term $\Psi(q_\mathrm{ch})$ defined in \eqref{eq:Psi}. 

From the above developments, we observe that the overall flow and volume dynamics of the DH system (in closed-loop) is given by   $\Sigma_\mathrm{ch} \circ \Sigma_\mathrm{pr,sh}$. Since  $\Sigma_\mathrm{ch}$ is independent of the states of $\Sigma_\mathrm{pr,sh}$, then these subsystems are in cascade as we show in Fig.~\ref{fig:4}. We move on to show that $\Sigma_\mathrm{ch} \circ \Sigma_\mathrm{pr,sh}$ has a unique equilibrium point.  It was shown in Proposition~\ref{prop:stab_q_ch} that $(\bar{q}_\mathrm{ch},\bar{x}_\mathrm{ch})=(q_\mathrm{ch}^\star,-f_\mathrm{ch}(q_\mathrm{ch}^\star))$ is the unique equilibrium of  $\Sigma_\mathrm{ch}$. For $\Sigma_\mathrm{pr,sh}$, we see from \eqref{eq:Sigma_prsh} that if $q_\mathrm{ch}=\bar{q}_\mathrm{ch}$, which implies $\Psi(q_\mathrm{ch})=0$,  then the vector $\bar{X}$, as given in \eqref{eq:bar X}, is a unique equilibrium of $\Sigma_\mathrm{pr,sh}$, provided that ${F}(\bar{X})$ is non singular. In the present case, non singularity of ${F}(\bar{X})$  is guaranteed from the assumption  $\bar{q}_\mathrm{pr}\neq 0$ (see equation~\eqref{eq:Sigma_prsh_prime} and Remark~\ref{rem:PE}).  It follows that $(\bar{q}_\mathrm{ch},\bar{x}_\mathrm{ch},\bar{X})$ is a unique equilibrium point of $\Sigma_\mathrm{ch} \circ \Sigma_\mathrm{pr,sh}$. 

To see that $(\bar{q}_\mathrm{ch},\bar{x}_\mathrm{ch},\bar{X})$ is (locally) asymptotically stable, we note on the one hand that $(\bar{q}_\mathrm{ch},\bar{x}_\mathrm{ch})$ is globally asymptotically stable for $\Sigma_\mathrm{ch}$ (see Proposition~\ref{prop:stab_q_ch}).  On the other hand, it was established in Proposition~\ref{prop:stab_vsh} that if $q_\mathrm{ch}=\bar{q}_\mathrm{ch}$ ($\Rightarrow \Psi(q_\mathrm{ch})=0$), then $\bar X$ is globally asymptotically stable for $\Sigma_\mathrm{pr,sh}$. Thus, we can invoke \cite[Proposition~4.1]{sepulchre_constructive_book} to conclude that the overall (coupled) system $\Sigma_\mathrm{ch} \circ \Sigma_\mathrm{pr,sh}$,  with $\Psi(q_\mathrm{ch})$ acting as an exogenous {\em vanishing} disturbance on $\Sigma_\mathrm{pr,sh}$,    admits $(\bar{q}_\mathrm{ch},\bar{x}_\mathrm{ch},\bar{X})$ as a unique, locally asymptotically stable equilibrium point. Moreover,  $\lim_{t\rightarrow \infty} q_\mathrm{ch}=q^\star_\mathrm{ch}$ and $\lim_{t\rightarrow \infty}V_\mathrm{sh}=V^\star_\mathrm{sh}$, provided that the system's initial conditions are sufficiently close to the equilibrium.
}
\end{proof}

\begin{figure}[h]
\begin{center}
\includegraphics[width=0.55\linewidth]{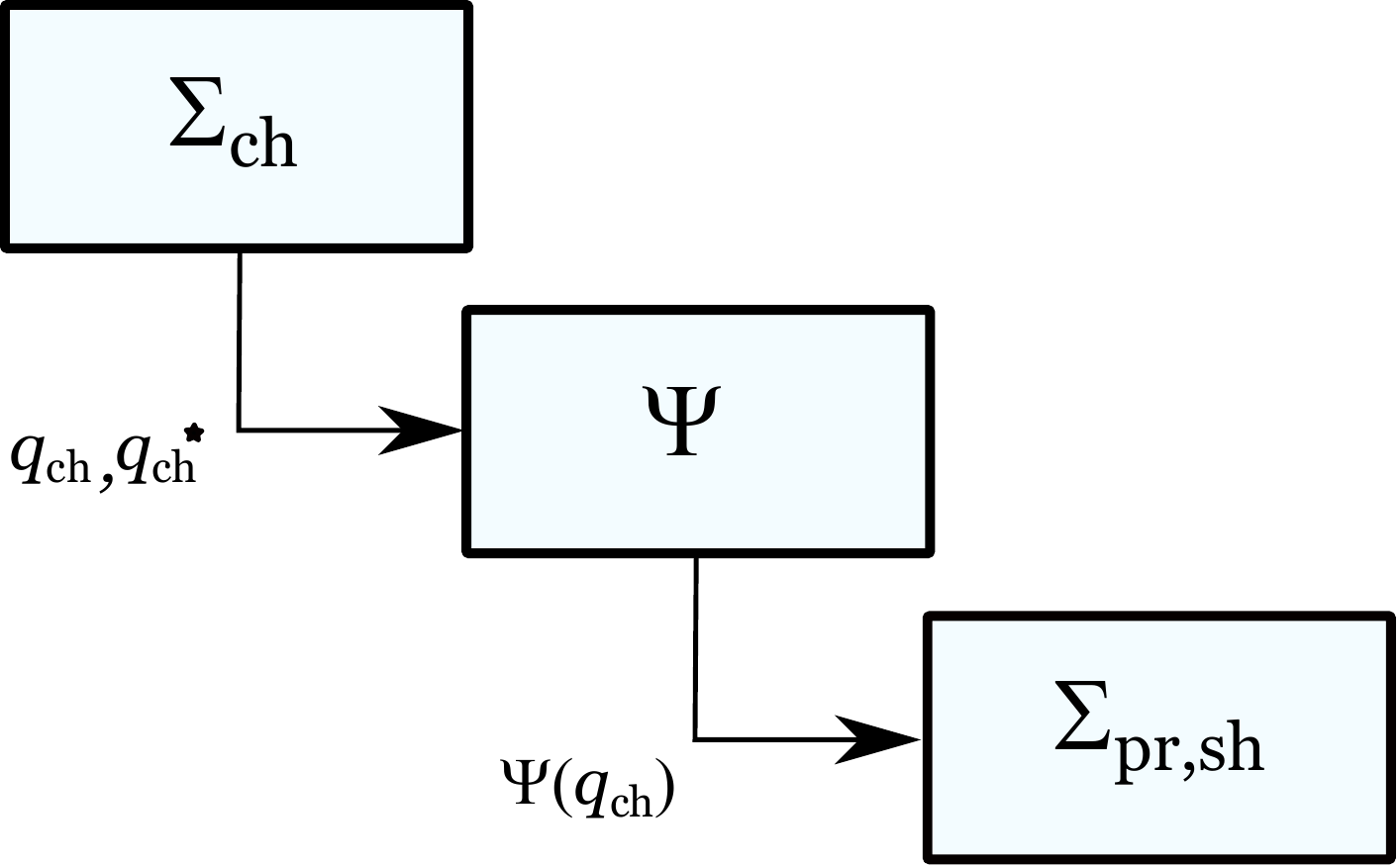}
\caption{Cascade interconnection between $\Sigma_\mathrm{ch}$ in \eqref{eq:Sigma_ch} and $\Sigma_\mathrm{pr,sh}$ in \eqref{eq:Sigma_prsh_prime} through the linear mapping $\Psi$ in \eqref{eq:Psi}.}
\label{fig:4}
\end{center}
\end{figure}

\begin{remark}
In the preceding proof, the subsystems  $\Sigma_\mathrm{ch}$ and $\Sigma_\mathrm{pr,sh}$ were shown to be globally asymptotically stable if they are decoupled, i.e., if $\Psi(q_\mathrm{ch})=0$ for all time. Notwithstanding, the result \cite[Proposition~4.1]{sepulchre_constructive_book} allows us only to claim local stability of the overall coupled system. In view of this drawback, part of our current research efforts are aimed at providing estimates of the system's domain of attraction. 
\end{remark}

\begin{remark}
Since the matrices $M_\mathrm{ch}$, $N_\mathrm{ch}$, $N_\mathrm{pr}$,  $N_\mathrm{sh}$, $M_\mathrm{a}$ and $M_\mathrm{b}$ are all diagonal, then  dynamic controllers \eqref{eq:control_q_ch_1} and \eqref{eq:controller_qpr_vsh} are fully decentralized.
\end{remark}

\section{Numerical Simulations}\label{sec:4}

In this section the performance of the DH system model in closed-loop with the proposed controller is illustrated via numerical simulations. We have used the configuration and data of the case study  reported in \cite[Section~4]{jmc_dh_modeling_2020}, which corresponds to a DH system with three heat producers ($n_\mathrm{pr}=3$),  nine consumers ($n_c=9$) and with the same topology as the sketch shown in Fig.~\ref{fig:1}. Thus, $n_\mathrm{ch}=17$. All producers are interfaced to the distribution network through storage tanks; each tank is assumed to have a total capacity of $1000~\mathrm{m}^3$.


The tuning gains of the dynamic controllers \eqref{eq:control_q_ch_1} and \eqref{eq:controller_qpr_vsh} are taken as $M_\mathrm{ch}=N_\mathrm{ch}=10^5 I_{n_\mathrm{ch}}$, and $N_\mathrm{pr}=7.11\times 10^4I_{n_\mathrm{pr}}$, $N_\mathrm{sh}=7.5\times 10^{-3} I_{n_\mathrm{pr}}$, $M_\mathrm{a}=14.06\times 10^{-5} I_{n_\mathrm{pr}}$ and $M_\mathrm{b}=7.11\times 10^{7} I_{n_\mathrm{pr}}$, respectively. This selection is based on a trial-and-error procedure aimed at attaining a fair balance between settling time and overshoot for the signals of interest. We note that, with the purpose of keeping the entries of $q_\mathrm{pr}$ within a sensible domain and to avoid, for example, flow reversals or too high flow rates, we have clipped the components of $u_\mathrm{pr}$ so each of them lie within 3\% and 115\% of a given  nominal equilibrium value of $u_{\mathrm{pr}}$ at maximum consumer demand.

We provide now a detailed explanation of the  simulation results that are shown in Figs. \ref{fig:sims_1} and \ref{fig:sims_2}. The system is initialized in the vicinity of system's equilibrium representing a context of low consumer demand  (25\% w.r.t. full demand) and with relatively small setpoints for each component of $V_\mathrm{sh}$. For simplicity we are assuming that the consumers' heat demand is proportional to their flow setpoints. Convergence is observed after a short transient. At $t=6$h all storage tanks switch to a charging mode and attain their respective desired  volume at approximately $t=9$h. The tanks switching to a charging mode causes  an increase in the producers'  pump actuation $u_\mathrm{pr}$ through the duration of the process, after which $u_\mathrm{pr}$ returns to its associated equilibrium value. At $t=12$h the reference value $q_{ch}^\star$ is changed to represent a context of high consumer demand (from 25\% to 95\% w.r.t. full demand). The plot of $q_{ch}$ shows that  convergence is achieved relatively quickly. Moreover, there is no significantly large overshoot for the consumers' pump actuation $u_\mathrm{ch}$ relative to its new equilibrium value. We note that the change in $q_\mathrm{ch}^\star$ induces a peak in the entries of $u_\mathrm{pr}$ but they return to their equilibrium values after a short time. At $t=18$h the tanks switch now to a discharging mode that ends at  approximately $t=21$h. During this process, it is possible to see a reduction in producers' pump actuation, contrary to what is observed during the tanks' charging mode. The new, lower values for the entries of $V_\mathrm{sh}^\star$ are maintained until the end of the simulation.


\begin{figure}
\begin{subfigure}{\linewidth}
\includegraphics[width=\textwidth]{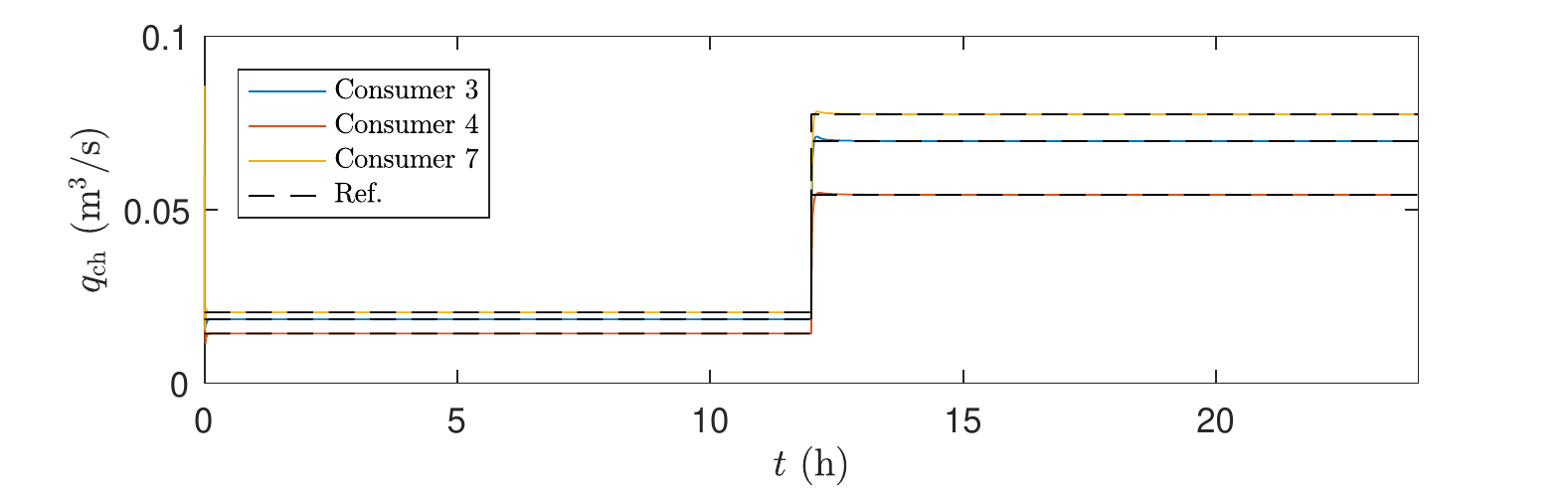}
\end{subfigure}
\begin{subfigure}{\linewidth}
\includegraphics[width=\textwidth]{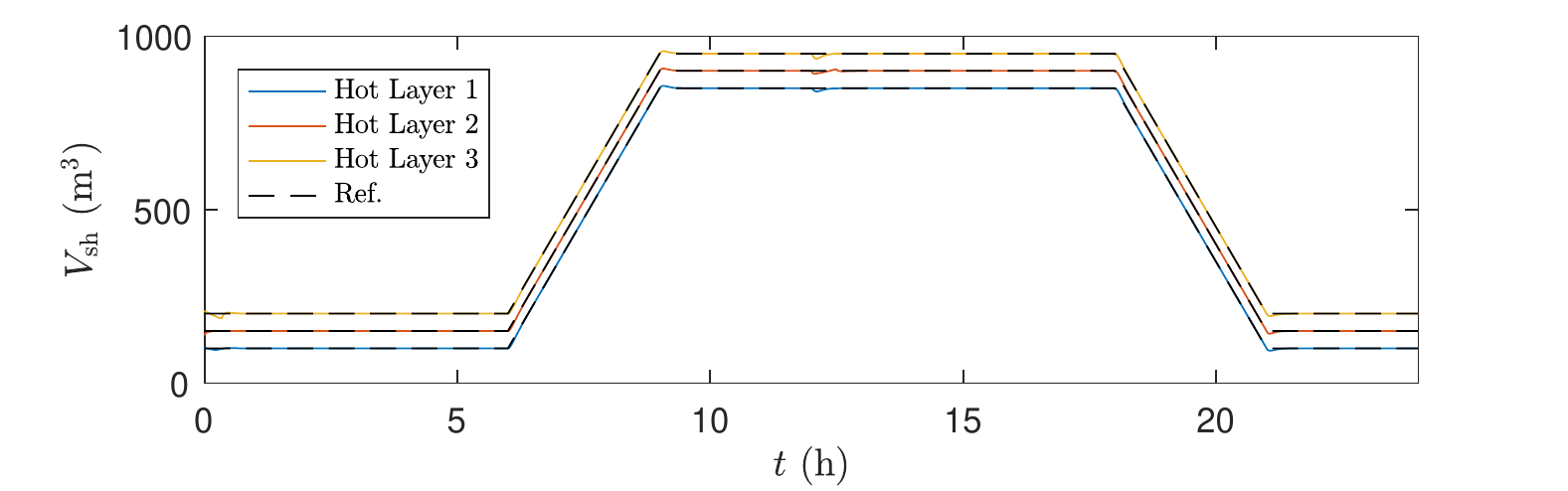}
\end{subfigure}
\caption{Evolution of the flow vector $q_\mathrm{ch}$ (top) and of the volume of hot water in the storage tanks (bottom).}
\label{fig:sims_1}
\end{figure}

\begin{figure}
\begin{subfigure}{\linewidth}
\includegraphics[width=\textwidth]{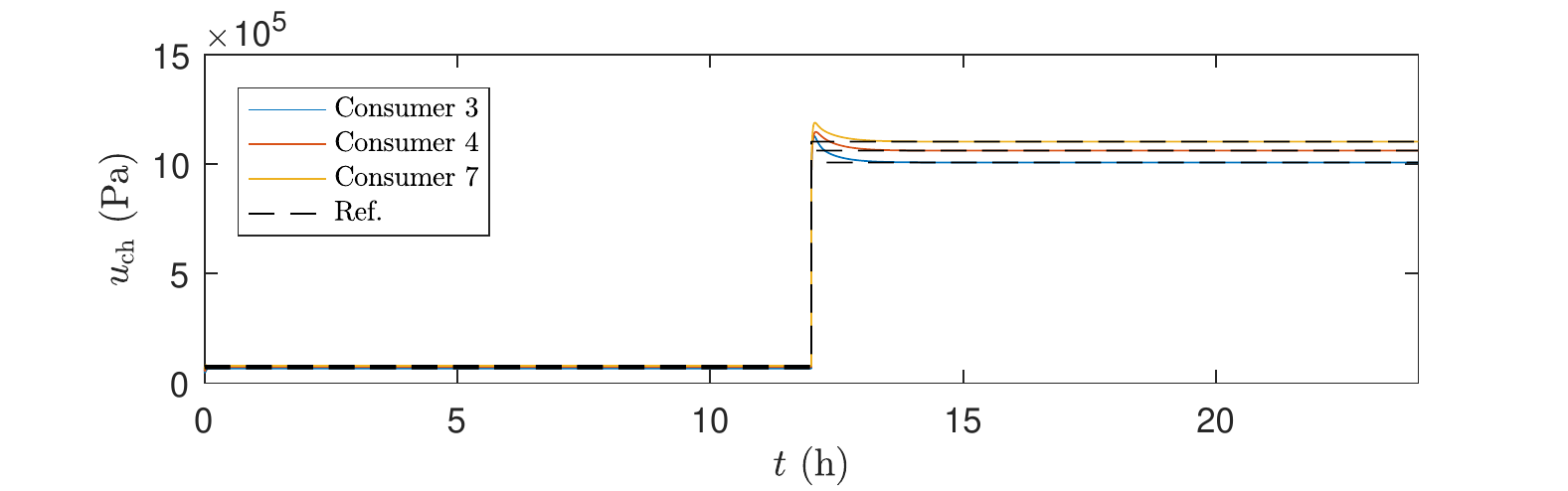}
\end{subfigure}
\begin{subfigure}{\linewidth}
\includegraphics[width=\textwidth]{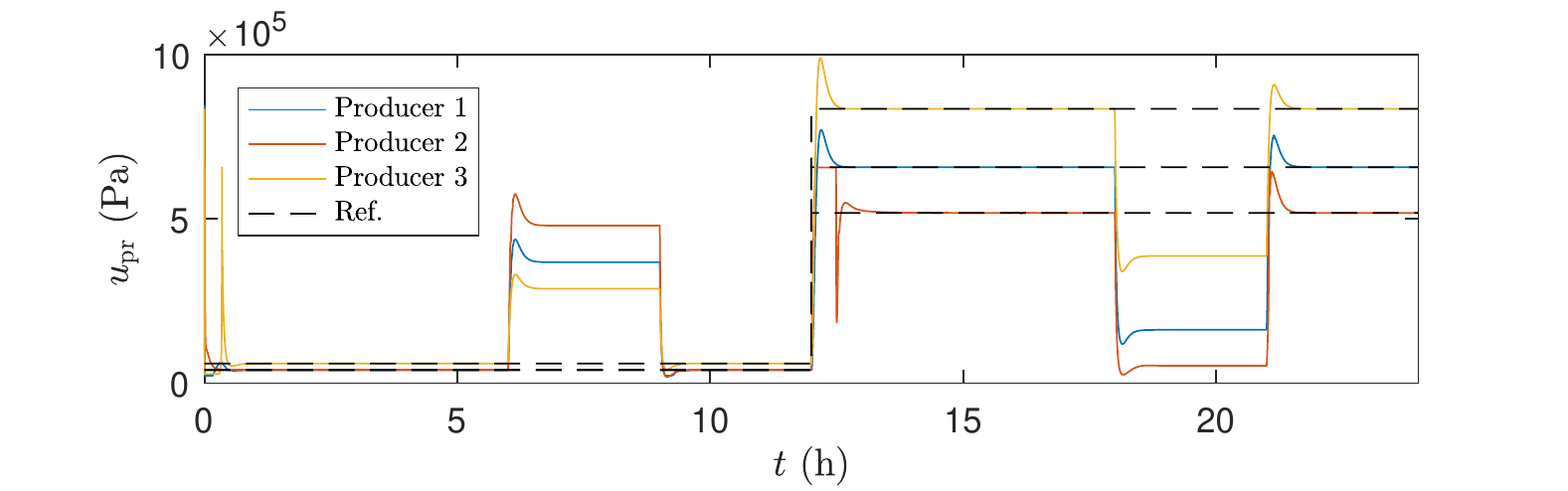}
\end{subfigure}
\caption{Evolution of the DH system's control inputs $u_\mathrm{ch}$ (top) and $u_\mathrm{pr}$ (bottom).}
\label{fig:sims_2}
\end{figure}

\section{Conclusions}\label{sec:5}

In this work we have addressed the flow and storage volume regulation of a multi-producer DH system via a novel adaptive decentralized control scheme which offers closed-loop (local) stability guarantees and overcomes the nonlinear, networked and uncertain characteristics of the considered system model, {\color{black}notably stemming from our consideration of frictional effects in pipes.} Our theoretical findings have been satisfactorily supported by simulation results on a realistic case study. Part of our ongoing research is related to the establishment of estimates of the closed-loop system's domain of attraction, {\color{black} the identification of gain tuning rules to adjust diverse performance criteria, e.g., the speed of convergence of $x_\mathrm{b}$ towards $\theta$, and the compatibility of our control design procedure (and closed-loop stability analysis) with other, more general, parameter and disturbance estimation schemes  (see, e.g., \cite{kolmanovsky2006simultaneous}), to formally study and address the effects of measurement noise, which were not considered in the present work}.

\AtNextBibliography{\small}
\printbibliography

@article{kolmanovsky2006simultaneous,
  title={Simultaneous input and parameter estimation with input observers and set-membership parameter bounding: theory and an automotive application},
  author={Kolmanovsky, Ilya and Sivergina, Irina and Sun, Jing},
  journal={International Journal of Adaptive Control and Signal Processing},
  volume={20},
  number={5},
  pages={225--246},
  year={2006}
}

@book{khalil_book_NLCONTROL_GLOBAL,
  title={Nonlinear Control, Global Edition},
  author={Khalil, H.},
  year={2015},
  publisher={Pearson}
}

@article{vesterlund_optim_2017,
  title={Optimization of multi-source complex district heating network, a case study},
  author={Vesterlund, M. and Toffolo, A. and Dahl, J.},
  journal={Energy},
  volume={126},
  pages={53--63},
  year={2017},
  publisher={Elsevier}
}

@article{wang_optimization_2017,
  title={Optimization modeling for smart operation of multi-source district heating with distributed variable-speed pumps},
  author={Wang, H. and Wang, H. and Haijian, Z. and Zhu, T.},
  journal={Energy},
  volume={138},
  pages={1247--1262},
  year={2017},
  publisher={Elsevier}
}

@article{bendtsen_control_2017,
  title={Control of district heating system with flow-dependent delays},
  author={Bendtsen, J. and Val, J. and Kalles{\o}e, C. and Krstic, M.},
  journal={IFAC-PapersOnLine},
  volume={50},
  number={1},
  pages={13612--13617},
  year={2017},
  publisher={Elsevier}
}

@article{ortega_garcia_04,
  title={Interconnection and damping assignment passivity-based control: A survey},
  author={Ortega, R. and Garcia-Canseco, E.},
  journal={European Journal of control},
  volume={10},
  number={5},
  pages={432--450},
  year={2004},
  publisher={Elsevier}
}

@article{nageshrao_ph_adaptive_16,
author = {Nageshrao, Subramanya P. and Lopes, Gabriel A.D. and Jeltsema, Dimitri and Babu{\v{s}}ka, Robert},
doi = {10.1109/TAC.2015.2458491},
issn = {00189286},
journal = {IEEE Transactions on Automatic Control},
number = {5},
pages = {1223--1238},
publisher = {IEEE},
title = {{Port-Hamiltonian Systems in Adaptive and Learning Control: A Survey}},
volume = {61},
year = {2016}
}

@book{sepulchre_constructive_book,
author = {Sepulchre, Rodolphe and Jankov{\'{i}}c, Mrdjan and Kokotovic, Petar},
publisher = {Springer},
title = {{Constructive Nonlinear Control}},
year = {1997}
}

@article{verda_storage_11,
author = {Verda, Vittorio and Colella, Francesco},
doi = {10.1016/j.energy.2011.04.015},
issn = {03605442},
journal = {Energy},
number = {7},
pages = {4278--4286},
publisher = {Elsevier Ltd},
title = {{Primary energy savings through thermal storage in district heating networks}},
url = {http://dx.doi.org/10.1016/j.energy.2011.04.015},
volume = {36},
year = {2011}
}

@article{jmc_dh_modeling_2020,
author={Machado, Juan E. and Cucuzzella, Michele and Scherpen, Jacquelien M. A.},
  title = {Modeling and Passivity Properties of District Heating Systems},
  journal={arXiv preprint arXiv:2011.05419},
  year={2021},
}

@article{valdimarsson_14,
author = {Vladimarsson, Pall},
doi = {10.4039/Ent30239-9},
issn = {19183240},
journal = {United Nations University Geothermal Training Programme},
number = {9},
pages = {239--240},
title = {{District Heat Distribution Networks}},
volume = {30},
year = {1898}
}

@incollection{Hauschild2020,
  title={Port-Hamiltonian modeling of district heating networks},
  author={Hauschild, Sarah-Alexa and Marheineke, Nicole and Mehrmann, Volker and Mohring, Jan and Badlyan, Arbi Moses and Rein, Markus and Schmidt, Martin},
  booktitle={Progress in Differential-Algebraic Equations II},
  pages={333--355},
  year={2020},
  publisher={Springer}
}

@article{Werner2017,
author = {Werner, Sven},
doi = {10.1016/j.energy.2017.04.045},
issn = {03605442},
journal = {Energy},
pages = {617--631},
publisher = {Elsevier Ltd},
title = {{International review of district heating and cooling}},
url = {https://doi.org/10.1016/j.energy.2017.04.045},
volume = {137},
year = {2017}
}

@article{kamal_storage_97,
author = {Ismail, Kamal and Leal, Jana{\'{i}}na and Zanardi, Maur{\'{i}}cio},
doi = {10.1016/S0360-5442(96)00172-7},
issn = {03605442},
journal = {Energy},
number = {8},
pages = {805--815},
title = {{Models of liquid storage tanks}},
volume = {22},
year = {1997}
}

@article{sandou_predictive_05,
author = {Sandou, G. and Font, S. and Tebbani, S. and Hiret, A. and Mondon, C.},
doi = {10.1109/CDC.2005.1583351},
isbn = {0780395689},
journal = {Proceedings of the 44th IEEE Conference on Decision and Control, and the European Control Conference, CDC-ECC '05},
pages = {7372--7377},
title = {{Predictive control of a complex district heating network}},
volume = {2005},
year = {2005}
}

@book{palsson_book_99,
author = {P{\'{a}}lsson, Hald{\'{o}}r and Larsen, Helge and Bohm, Benny and Ravn, Hans and Zhou, Jijun},
publisher = {Technical University of Denmark},
title = {{Equivalent models of district heating systems}},
year = {1999}
}

@article{Scholten_tcst_2015,
author = {Scholten, T. and De Persis, C. and Tesi, P.},
doi = {10.1109/ECC.2015.7330872},
isbn = {9783952426937},
journal = {IEEE Transactions on Control Systems Technology},
number = {2},
pages = {414----427},
title = {{Modeling and control of heat networks with storage: The single-producer multiple-consumer case}},
volume = {25},
year = {2015}
}

@article{wang_meshed_17,
author = {Wang, Yaran and You, Shijun and Zhang, Huan and Zheng, Wandong and Zheng, Xuejing and Miao, Qingwei},
doi = {10.1016/j.energy.2017.03.044},
issn = {03605442},
journal = {Energy},
pages = {603--621},
publisher = {Elsevier Ltd},
title = {{Hydraulic performance optimization of meshed district heating network with multiple heat sources}},
url = {http://dx.doi.org/10.1016/j.energy.2017.03.044},
volume = {126},
year = {2017}
}

@article{Dominkovic2017,
author = {Dominkovi{\'{c}}, D. F. and Ba{\v{c}}ekovi{\'{c}}, I. and Sveinbj{\"{o}}rnsson, D. and Pedersen, A. S. and Kraja{\v{c}}i{\'{c}}, G.},
doi = {10.1016/j.energy.2017.02.162},
issn = {03605442},
journal = {Energy},
pages = {941--960},
title = {{On the way towards smart energy supply in cities: The impact of interconnecting geographically distributed district heating grids on the energy system}},
volume = {137},
year = {2017}
}

@article{Lund2014,
author = {Lund, Henrik and Werner, Sven and Wiltshire, Robin and Svendsen, Svend and Thorsen, Jan Eric and Hvelplund, Frede and Mathiesen, Brian Vad},
doi = {10.1016/j.energy.2014.02.089},
issn = {03605442},
journal = {Energy},
pages = {1--11},
publisher = {Elsevier Ltd},
title = {{4th Generation District Heating (4GDH). Integrating smart thermal grids into future sustainable energy systems.}},
url = {http://dx.doi.org/10.1016/j.energy.2014.02.089},
volume = {68},
year = {2014}
}

@book{cengel_thermo_08,
author = {Cengel, Yunus A.},
edition = {2nd},
publisher = {McGraw Hill},
title = {{Introduction to Thermodynamics and Heat Transfer}}
}

@article{vandermeulen_control_review_18,
author = {Vandermeulen, Annelies and van der Heijde, Bram and Helsen, Lieve},
doi = {10.1016/j.energy.2018.03.034},
issn = {0360-5442},
journal = {Energy},
pages = {103--115},
publisher = {Elsevier Ltd},
title = {{Controlling district heating and cooling networks to unlock flexibility : A review}},
url = {https://doi.org/10.1016/j.energy.2018.03.034},
volume = {151},
year = {2018}
}

@article{DePersis2014,
author = {{De Persis}, Claudio and Jensen, Tom and Ortega, Romeo and Wisniewski, Rafa{\l}},
doi = {10.1109/TCST.2012.2233477},
issn = {10636536},
journal = {IEEE Transactions on Control Systems Technology},
number = {1},
pages = {238--245},
title = {{Output regulation of large-scale hydraulic networks}},
volume = {22},
year = {2014}
}

@article{nima_scl_19,
author = {Monshizadeh, Nima and Monshizadeh, Pooya and Ortega, Romeo and van der Schaft, Arjan},
journal = {Systems and Control Letters},
pages = {55--61},
publisher = {Elsevier B.V.},
title = {{Conditions on shifted passivity of port-Hamiltonian systems}},
url = {https://doi.org/10.1016/j.sysconle.2018.10.010},
volume = {123},
year = {2019}
}

@article{bayu_scl_07,
author = {Jayawardhana, Bayu and Ortega, Romeo and Garc{\'{i}}a-Canseco, Elo{\'{i}}sa and Casta{\~{n}}os, Fernando},
doi = {10.1016/j.sysconle.2007.03.011},
issn = {01676911},
journal = {Systems and Control Letters},
pages = {618--622},
title = {{Passivity of nonlinear incremental systems: Application to PI stabilization of nonlinear RLC circuits}},
volume = {56},
year = {2007}
}

@article{Trip2019a,
author = {Trip, Sebastian and Scholten, Tjardo and {De Persis}, Claudio},
doi = {10.1016/j.automatica.2019.02.046},
issn = {00051098},
journal = {Automatica},
pages = {141--153},
title = {{Optimal regulation of flow networks with transient constraints}},
volume = {104},
year = {2019}
}

@article{grosswindhager_delay_11,
author = {Grosswindhager, Stefan and Voigt, Andreas and Kozek, Martin},
doi = {10.2316/P.2011.735-094},
isbn = {9780889868878},
issn = {10218181},
journal = {Proceedings of the IASTED International Conference on Modelling and Simulation},
number = {2},
pages = {41--48},
title = {{Efficient physical modelling of district heating networks}},
year = {2011}
}

@article{DePersis2011,
author = {{De Persis}, Claudio and Kalles{\o}e, Carsten},
doi = {10.1109/TCST.2010.2094619},
issn = {10636536},
journal = {IEEE Transactions on Control Systems Technology},
number = {6},
pages = {1371--1383},
title = {{Pressure regulation in nonlinear hydraulic networks by positive and quantized controls}},
volume = {19},
year = {2011}
}

\end{document}